\documentclass[10pt,journal,twocolumn]{IEEEtran}
\ifCLASSINFOpdf
\else
\fi
\usepackage{graphicx}
\usepackage{graphics} 
\usepackage{epsfig} 
\usepackage{epstopdf}
\usepackage{mathptmx} 
\usepackage{times} 
\usepackage{amsmath} 
\usepackage{amsthm} 
\usepackage{amssymb}  
\usepackage{cite}
\usepackage{color}
\usepackage{array}
\usepackage{color}
\usepackage{floatflt}
\usepackage{multicol}
\usepackage{lipsum}
\usepackage{amsfonts} 
\usepackage{lastpage}
\usepackage{latexsym}
\usepackage[ruled]{algorithm2e}

\newtheorem{definition}{Definition}

\newtheorem{remark}{\textit{Remark}}

\newtheorem{theorem}{\bf{Theorem}}

\newtheorem{proposition}{\bf{Proposition}}
\newtheorem{corollary}{\bf{Corollary}}

\begin{document}

\title{A Message Passing Strategy for Decentralized Connectivity Maintenance in Agent Removal}

%
%

\author{Derya~Aksaray,~
        A.~Yasin~Yaz{\i}c{\i}o\u{g}lu,~
        Eric~Feron, and~Dimitri~N.~Mavris~
\thanks{D.~Aksaray,~E.~Feron,~and D.N.~Mavris are with the Department of Aerospace Engineering, Georgia Institute of Technology, Atlanta,
GA, 30332 USA (e-mails: daksaray@gatech.edu, eric.feron@aerospace.gatech.edu, dimitri.mavris@aerospace@gatech.edu).}
\thanks{A.Y.~Yaz{\i}c{\i}o\u{g}lu is with the Department of Electrical and Computer Engineering, Georgia Institute of Technology, Atlanta,
GA, 30332 USA (e-mail: yasin@gatech.edu).}
}
\bibliographystyle{IEEEtran}
\maketitle

\begin{abstract}
In a multi-agent system, agents coordinate to achieve global tasks through local communications. Coordination usually requires sufficient information flow, which is usually depicted by the connectivity of the communication network. In a networked system, removal of some agents may cause a disconnection. In order to maintain connectivity in agent removal, one can design a robust network topology that tolerates a finite number of agent losses, and/or develop a control strategy that recovers connectivity. This paper proposes a decentralized control scheme based on a sequence of replacements, each of which occurs between an agent and one of its immediate neighbors. The replacements always end with an agent, whose relocation does not cause a disconnection. We show that such an agent can be reached by a local rule utilizing only some local information available in agents' immediate neighborhoods. As such, the proposed message passing strategy guarantees the connectivity maintenance in arbitrary agent removal. Furthermore, we significantly improve the optimality of the proposed scheme by incorporating $\delta$-criticality (i.e. the criticality of an agent in its $\delta$-neighborhood).
\end{abstract}
\maketitle

\section{Introduction} \label{intro}
OVER the last decade, advances in networking and computing technologies along with new manufacturing techniques have created a new paradigm shift towards multi-agent systems (MAS) in engineering applications. A MAS involves a set of agents. In most cases, some information flow among the agents via local communications. Recently, there is a significant interest in using multiple unmanned vehicles over large areas for target tracking (e.g. \cite{martinez2006, olfatisaber2007}), environmental monitoring (e.g. \cite{leonard2007, white2008}), persistent surveillance (e.g. \cite{beard2006, bethke2009}), formation and coverage (e.g. \cite{fax2004, li2010}), and several others. In these missions, it is often desirable to have a connected communication network. However, an agent removal may cause a disconnection. 

Multi-agent health management problems emphasize that there is an inherent possibility of agent removal in a MAS due to failure or refueling (e.g. \cite{redding2011, valenti2007,valenti2007_2}). In such cases, connectivity maintenance can be achieved through \emph{proactive} (e.g. \cite{motevallian2011, summers2009, yazicioglu2013}) or \emph{reactive} (e.g. \cite{akkaya2008, abbasi2013}) approaches. In proactive approaches, a robust network topology is designed a priori to mission such that the network can tolerate a finite number of agent losses. In reactive approaches, a control strategy is developed such that the network self-repairs itself in the case of agent removal. Note that relying only on proactive approaches, by adding redundant edges to the graph for strong connectivity, can be impractical in missions that have limited resources for communication and a possibility of losing large number of agents.

Recently, there is a growing interest in the recovery processes for connectivity maintenance of networked systems. These processes can be characterized as \textit{centralized} or \textit{decentralized} with respect to the available information of the overall or local network structure, respectively. In a large scale system, the availability of the overall network is not a realistic assumption. Therefore, a decentralized strategy is preferable over a centralized one due to practicality and scalability concerns. 

This paper introduces a decentralized recovery scheme that is applicable to any scale of networked systems. The decentralized scheme is based on a sequence of replacements occurring between an agent and one of its neighbors. Each agent is assumed to have a unique ID, which is known to its immediate neighbors. Before an agent leaves the group (e.g. due to reaching to a critical power threshold), it creates a message with its individual ID and passes it to one of its neighbors as a request for that neighbor to replace itself. Whenever an agent receives a message, it adds its own ID to the bottom of the message and sends it to one of its other neighbors, whose ID is not included in the message, i.e. a neighbor who has not received the message earlier. Accordingly, some consecutive replacements are executed. A message passing strategy is similar to token-based techniques used in various algorithms, such as in \cite{petcu2006} and \cite{corke2006}, to record the nodes visited by the token. In this paper, we use such a token-based idea, and we show that the resulting replacement sequences guarantee connectivity in any agent removal. The proposed strategy is a decentralized scheme that leverages only local information (e.g. agent IDs) for the connectivity maintenance, hence it is applicable to networks of any scale.   

The organization of this paper is as follows: Section~\ref{literature} presents some related work in the literature. Section~\ref{math_prelim} depicts some mathematical preliminaries. Section~\ref{problem} motivates and defines the problem. Sections~\ref{replacement_strategy} and \ref{deltaMPS} elaborate on the replacement control problem and introduce the message passing strategies. Section~\ref{analysis} presents the Monte Carlo simulations for the analysis of proposed control scheme and its comparison to an optimal (minimum number of replacements) centralized method. Finally, Section~\ref{conclusion} concludes the paper.

\section{Related Work} \label{literature}
Recently, a great amount of interest has been devoted to the analysis of multi-agent systems via the graph theory. In these studies, the \textit{nodes} of a graph represent the agents (such as robots, sensors or individuals), and  \textit{edges} represent the direct interactions between them. For such a representation, a fundamental graph property related to the system robustness is graph connectivity (e.g. \cite{ji2007, zavlanos2007, zavlanos2011} and the literature cited within). As such, the robustness of a system is related to the total number of edges/nodes, whose removal will cause a network disconnection. For the graph theoretical connectivity control of mobile systems against edge failure, the literature is including, but not limited to, optimization based connectivity control (e.g. \cite{zavlanos2011}), continuous feedback connectivity control (e.g. \cite{sabattini2011}), and control based on the estimation of the algebraic connectivity (e.g. \cite{sabattini2012}). In these studies, the authors mainly consider uncertainty in edges and assume a constant number of nodes.  

Maintaining connectivity against the removal of multiple agents is a more challenging problem than maintaining connectivity against multiple edge removal \cite{motevallian2011}. In the last few years, there has been a significant interest in addressing agent loss problem in networked systems. In \cite{motevallian2011} and \cite{summers2009}, the main focus is on the design of robust network topologies that can tolerate a finite number of agent removals. In \cite{summers2009} and \cite{bauso2012}, the authors propose self-repair strategies that create new connections among the neighbors of the failing agent. In addition, a connectivity maintenance strategy based on decentralized estimation of algebraic connectivity is presented in \cite{knorn2009}. Based on their estimations, agents increase or decrease their broadcast radii for satisfying connectivity requirements.

Different from the previous studies, \cite{akkaya2008}, \cite{abbasi2009}, \cite{wang2011} and \cite{abbasi2013} consider mobile agents and propose some agent movements for connectivity restoration of wireless sensor networks in agent failure. In \cite{akkaya2008}, a distributed control algorithm is introduced for connectivity maintenance. Before any failure, the algorithm runs and identifies all critical agents, whose failure will cause network disconnection. Then, it assigns required actions to each agent in advance. The studies in \cite{abbasi2009} and \cite{wang2011} differ from \cite{akkaya2008} by maintaining connectivity through some agent relocations initiated by the failing agent. In \cite{wang2011}, the authors present a centralized algorithm as an alternative to the decentralized scheme given in \cite{abbasi2009}, which is not always feasible in general graphs. Finally, the authors of \cite{abbasi2013} use the shortest path routing table in their algorithm, and they propose a distributed recovery mechanism that maintains the network connectivity with minimal topology change, i.e. not increasing the length of the shortest path between any arbitrary two agents after the reconfiguration. 


In this paper, we present a decentralized recovery mechanism to maintain network connectivity in arbitrary agent removal. The replacement control problem has been initially introduced in \cite{aksaray2013}, and replacements by minimum degree neighbors have been presented as a solution. Here, we generalize the connectivity maintenance scheme as the message passing strategy, and we show that this method maintains connectivity even in the case of agents sharing minimum amount of information, i.e. only node IDs. Moreover, we show that utilizing $\delta$-criticality information in the message passing strategy significantly improves the optimality of the solution.

\section{Mathematical Preliminaries} \label{math_prelim}
An undirected graph, $\mathcal{G}=(V,E)$, consists of a set of nodes, $V$, and a set of undirected edges, $E$. A $k$-length \textit{path}, $p$, is a sequence of nodes $(p_0,p_1,...,p_k)$ such that any $\{p_i,p_{i+1}\} \in E$. Here, $p_i$ is the $i^{th}$ element of $p$, which corresponds to a node $v_j$. A path in $\mathcal{G}$ is called \textit{simple} if it does not have any repeated nodes. The distance between any two nodes is equal to the length of the shortest path between them. The \emph{diameter} of $\mathcal{G}$, $\Delta$, is defined as the largest distance between any two nodes of $\mathcal{G}$. An undirected graph, $\mathcal{G}$, is \emph{connected} if there exists a path between any two nodes of the graph.
 
In a graph, the unweighted adjacency matrix, $\bf{A} \in \mathbb{R}^{\rm n \times n}$, is

\begin{eqnarray}
\bf{A}_{\rm ij} = \bigg\{
\begin{array}{l l}
    1 & \quad \text{if $(v_i,v_j)\in E$,}\\
    0 & \quad \text{otherwise}.\\
  \end{array} 
\end{eqnarray}

The neighbor set, $\mathcal{N}_{v_i}$, of $v_i$ is the set including all adjacent nodes that are connected to $v_i$. 

\begin{equation}
\mathcal{N}_{v_i} = \{v_j \mid (v_i,v_j) \in E \}.
\end{equation}

The \emph{degree} of $v_i$ is the number of nodes adjacent to $v_i$, in other words the cardinality of $\mathcal{N}_{v_i}$.


%
%

\section{Problem Formulation} \label{problem}
Given a networked system with no control strategy for the connectivity maintenance, the overall network will eventually become disconnected as more agents are removed. In fact, it is possible to observe graph disconnection after the removal of a small set of randomly selected agents. We demonstrate this claim by conducting Monte Carlo simulations with Erd\H{o}s-R\'{e}nyi graphs \cite{erdhos1960}. Let a connected graph, $\mathcal{G}$, have $50$ nodes that have a probability of $0.04$ to connect with other nodes. Suppose that a randomly selected node and its incident edges are iteratively removed until the occurrence of the first graph disconnection. When the graph is disconnected for the first time, we record the corresponding iteration number ($i_{dis}$). Then we repeat this process for $1000$ cases that are initiated with different random graphs. The results of the simulations are displayed as an empirical cumulative distribution function of $i_{dis}$ illustrated in Figure~\ref{fig:cdf}. To provide insight into the randomly generated graphs, Figure~\ref{fig:cdf} also displays the average degree distribution. As it is seen from the figure, the removal of $10$ nodes most likely cause disconnection in random graphs with $50$ nodes and an approximate average degree of $2.5$.   

\begin{figure}[h!]
\begin{center}
\resizebox*{8cm}{!}{\includegraphics{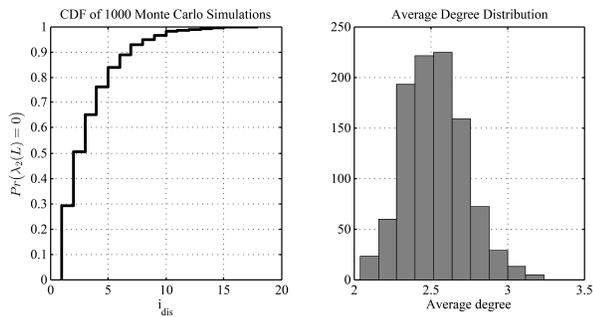}}%
\caption{(a) The empirical CDF for the iteration number when a graph disconnection is observed for the first time. (b) The average degree distribution for the generated Erd\H{o}s-R\'{e}nyi graphs in the Monte Carlo simulations.}%
\label{fig:cdf}
\end{center}
\end{figure}

As the degrees of nodes increase in a graph, the robustness to agent failure increases. However, there is a compromise between improving the robustness and increasing the total edge number. While a robust graph is tolerant to a finite number of agent removals, it involves many redundant edges that increase the overall communication cost of the network. Alternative to designing a network with a large number of redundant links, a control strategy can be utilized to maintain connectivity. Such a control strategy may provide self-reconfiguration on the graph whenever an agent removal occurs. For static networked systems, the removal of an agent will trigger to have new connections in the neighborhood of the removed agent. For mobile networked systems, the removal of an agent will induce some agent movements until the connectivity is maintained. 

In many distributed systems, network connectivity plays an important role in achieving a desired system performance. For instance, connectivity is required for the propagation of local data to achieve coordinated tasks. As such, formation of connected networks are emphasized in spacecraft studies (e.g. \cite{dai2012}, \cite{zhou2013}), where multiple spacecrafts synchronize their attitudes with each other via local interactions. Alternatively, a peer-to-peer network (e.g. Gnutella) is a distributed system, where individual computers communicate directly with each other and share information and resources without using centralized servers. In these systems, network protocols are designed to achieve various objectives, one of which is the ability to operate in a dynamic environment (e.g. \cite{pandurangan2003,ripeanu2002}). In this manner, when a host computer leaves a network, utilizing a connectivity maintenance strategy avoids network partitionings and prevents the performance degradation. Moreover, a group of sensors can be deployed in an area of interest to achieve distributed estimation in a harsh environment. In these systems, the network topology is dynamic due to the possibility of sensor failures or efficient energy management. As such, the sensors may increase or decrease their sensing radii to achieve desired connectivity requirements \cite{knorn2009}. Finally, connectivity is crucial in surveillance missions, where a group of heterogeneous agents or unmanned aerial vehicles (UAV) operate around a base. For example, an efficient task assignment, which also satisfies a connected communication network, provides agents to stream back the surveillance data back to the base \cite{ponda2011}. Similarly, a group of UAVs can monitor a desired region to track radar data and a base process the individually gathered data to estimate the position of a target \cite{casbeer2006}. In such problems, a UAV can go further away from the group to investigate unmonitored areas, or an agent can return to the base. The removal of an agent/UAV may cause a disconnection in the communication network, which  leads to a base not to collect data from the disconnected agent(s)/UAV(s).    

In this study, we consider an undirected connected graph, $\mathcal{G}$, where the nodes of $\mathcal{G}$ represent the feasible points that agents can be assigned to. Here, a feasible point can be abstract (tasks among peers in a computer network) or physical (areas of interest in robotic networks). Furthermore, the edges of $\mathcal{G}$ correspond the required interactions if a pair of agents are assigned on the corresponding nodes (tasks). In this setting, the communication network of agents is the sub-graph of $\mathcal{G}$ based on the agent assignments. Here, our main objective is to maintain the connectivity of the communication network in the arbitrary agent removal by properly assigning agents to the feasible points. One way to maintain the network connectivity in these problems is to replace the removed agent by one of the remaining agents. For instance, if the removal of an agent causes disconnection in the network, then one of its neighbors may replace it to recover the connectivity. If the replacement also causes a disconnection, then another replacement is also required. In this manner, the replacements can be executed until a connected network is obtained. Accordingly, we introduce the $\textit{replacement control problem}$ as follows:  \\

\noindent\textbf{Replacement Control (RC) problem}: \textit{Given a set of agents with a connected communication network, design a decentralized control scheme such that the agents realize minimum number of node replacements to maintain connectivity in the presence of agent removal.}

\section{Message Passing Strategy} \label{replacement_strategy}
For any solution of the RC problem, the sequence of replacements needs to end with a noncritical node since the removal of such nodes does not require any replacements. 
\begin{definition} (Node Criticality)
A node, $v_i$, is critical in $\mathcal{G}$ if the graph, $\mathcal{G}^* = \mathcal{G} - (v_i , E_i)$, obtained by removing $v_i$ and $E_i$ is disconnected; otherwise, $v_i$ is noncritical.
\label{def: criticality}
\end{definition}

Note that a connected graph always has a finite number of noncritical nodes \cite{savchenko2006}.

\begin{proposition}
\textit{(Existence of noncritical nodes)}: \cite{savchenko2006} Let $\mathcal{G}$ be a connected undirected graph. Suppose that each of its nodes has a degree at least $k$. Then $\mathcal{G}$ has at least $k + 1$ noncritical nodes. 
\label{prop: noncritical_existence}
\end{proposition}

For any connected graph, there are always at least two noncritical nodes, and the goal of the RC problem is to find one. The following remark presents a condition for a trivial noncritical node in a graph.

\begin{remark}
Given a connected graph $\mathcal{G} = (V,E)$, let $v_i \in V$ be a leaf node such that $|\mathcal{N}_{v_i}| = 1$. Then, $v_i$ is noncritical in $\mathcal{G}$ because any simple path involving $v_i$ either starts or ends with $v_i$. Hence, its removal will not cause a disconnection between any two nodes.
\label{rem: leaf}
\end{remark}

In the RC problem, a replacement is assumed to occur in between a node and one of its neighbors. Therefore, the sequence of replacements can also be defined as a path from the removed node to a noncritical node. 

\begin{remark}
Let $\mathcal{G}$ be a connected undirected graph. It follows from Proposition~\ref{prop: noncritical_existence} that there always exists a path from any node in $\mathcal{G}$ to a noncritical node.
\label{path_existence}
\end{remark}

Note that a centralized controller can solve the RC problem by finding a shortest path between the removed node and a noncritical node. Here, the optimal solution is obtained by assuming the availability of the overall graph structure. The goal driving this work is to find a decentralized scheme that can perform close to optimal. 

\begin{definition} (Maximal simple path)
Let $\mathcal{G}=(V,E)$ be a connected undirected graph, and let $\mathcal{N}_{v_i}$ denote the neighbors of $v_i \in V$. Suppose that $p=(p_0,p_1,...,p_k)$ is a simple path with a length of $k$. Then $p$ is a maximal simple path if $\mathcal{N}_{p_k} \subseteq \{p_0,p_1,...,p_k\}$.
\label{def: max_simplepath}
\end{definition}

\begin{theorem} \cite{aksaray2013}
Given a connected undirected graph $\mathcal{G}$, a maximal simple path on $\mathcal{G}$ always ends with a noncritical node.
\label{thm: mainThm}
\end{theorem}

\begin{corollary}
A sequence of replacements along a maximal simple path, $(p_0,p_1,...,p_k)$, on $\mathcal{G}$, such that $p_0$ represents any arbitrary removed node and any $p_{i+1} \in \mathcal{N}_{p_i} \setminus \{p_0,p_1,...,p_i\}$, maintains the graph connectivity.  
\label{cor: guarantee_connectivity}
\end{corollary}

\begin{proof}
The maximal simple path $(p_0,p_1,...,p_k)$ is the replacement path where $p_0$ is any arbitrary removed node and for $0 \leq i \leq k-1$ any $p_i$ is replaced by $p_{i+1}$. After the replacements are realized, the graph will have a new structure as if $p_k$ is removed from the system. From Theorem~\ref{thm: mainThm}, we know that $p_k$ is noncritical so its removal does not cause any disconnection in $\mathcal{G}$. 
\end{proof}

In light of the preceding facts, we introduce a decentralized connectivity maintenance scheme called \emph{message passing strategy} (MPS). Let $p_0$ be any arbitrary node that will be removed from $\mathcal{G}$. The objective of MPS is to find a sequence of replacements, which is initiated by $p_0$ and ending with a noncritical node, by using only some local information. In this manner, the replacements will result in a graph reconfiguration as if the final node in the replacement sequence, which is noncritical, is removed from $\mathcal{G}$ instead of $p_0$. 

The outline of MPS is as follows: Before the removal of $p_0$, first $p_0$ creates a message including its own node ID as $\{p_0\}$ and checks whether it is a leaf node. If it is a leaf node, then it is noncritical (from Remark~\ref{rem: leaf}) and its removal will not cause a disconnection. Otherwise, it selects a node, $p_1$, from $\mathcal{N}_{p_0} \setminus \{p_0\}$. Then, $p_0$ sends the message to $p_1$, which will replace $p_0$. In this respect, whenever a node, $p_i$, receives a message, $\{p_0,...,p_{i-1}\}$, from $p_{i-1}$, before $p_i$ replaces $p_{i-1}$, it adds its individual node ID to the bottom of the message as $\{p_0,...,p_{i-1},{\bf p_i}\}$, and it sends the message to one of its neighbors from the set $\mathcal{N}_{p_i} \setminus \{p_0,...,p_i\}$. Eventually, the message passing process, whose pseudo-code is displayed in Algorithm~1, stops when $\mathcal{N}_{p_i} \setminus \{p_0,...,p_i\} = \emptyset$ or $p_0$ is a leaf node.

\begin{center}
\resizebox{\columnwidth}{!}{
\begin{tabular}{l l} 
\hline
\bf{Algorithm 1: Message Passing Strategy (MPS)}\\
\hline
$Input :$ An arbitrary node, $p_0$, from $\mathcal{G}$  \\
$Output :$ Connectivity maintenance in the removal of $p_0$ \\
$Assumption:$ Each node shares its unique node ID with its neighbors.\\
\mbox{\small $\;1:\;$}\textbf{initialization:} $p_i \gets p_0; \quad \quad \mathcal{N}_{p_i} \gets \mathcal{N}_{p_0}; \quad$ $message \gets (p_0)$; \\
\mbox{\small $\;2:\;$}\textbf{if}\hspace{0.1cm}  $| \mathcal{N}_{p_0} |=1$ \\
\mbox{\small $\;3:\;$}\hspace{0.45cm} no replacements required; \\
\mbox{\small $\;4:\;$}\textbf{else}\hspace{0.1cm}  $ $ \\
\mbox{\small $\;5:\;$}\hspace{0.45cm} \textbf{while} \hspace{0.1cm} $\mathcal{N}_{p_i} \setminus message \neq \emptyset$\\
\mbox{\small $\;7:\;$}\hspace{0.90cm} $p_{i+1}$ $\gets$ $v \quad s.t. \quad v \in \mathcal{N}_{p_i} \setminus message$; \\
\mbox{\small $\;8:\;$}\hspace{0.90cm} $p_i$ sends $message$ to $p_{i+1}$;\\
\mbox{\small $\;9:\;$}\hspace{0.90cm} $p_i$ replaces the second last node in the $message$; \\
\mbox{\small $\;10:\;$}\hspace{0.75cm} $p_i \gets p_{i+1}$; \quad $\mathcal{N}_{p_i} \gets \mathcal{N}_{p_{i+1}}$; \\ 	
\mbox{\small $\;6:\;$}\hspace{0.90cm} $message \gets (message , p_i)$ ;\\
\mbox{\small $\;11:\;$}\hspace{0.45cm} \textbf{end while} \\
\mbox{\small $\;12:\;$}\textbf{end if}\hspace{0.1cm}  $ $ \\
\hline
\end{tabular}
}
\end{center}
\vskip1ex

\begin{proposition}
The message obtained from MPS results in a set of ordered nodes, which represents either a leaf node or a maximal simple path.
\label{prop: MPS_max_simple_path}
\end{proposition}
\begin{proof}
The message obtained from MPS is either $\{ p_0 \}$ or $\{p_0,...,p_i,...,p_k\}$. If it is $\{p_0\}$, then $| \mathcal{N}_{p_0} |=1$ implying that $p_0$ is a leaf node. If the message is $\{p_0,...,p_i,...,p_k\}$, it involves consecutive pairs of nodes, $(p_i,p_{i+1}) \in E$, thus the message always represents a path in $\mathcal{G}$. Additionally, the message never involves repeated nodes because each $p_i$ selects $p_{i+1}$ from $\mathcal{N}_{p_i} \setminus \{p_0,...,p_i\}$. Thus, the path is always simple. Finally, MPS stops whenever $\mathcal{N}_{p_k} \setminus \{p_0,...,p_i,...,p_k\} = \emptyset$. From Definition~\ref{def: max_simplepath}, the ordered nodes in the message is a maximal simple path.
\end{proof}

\begin{corollary}
MPS always stops at a noncritical node. Hence, MPS guarantees connectivity maintenance in the removal of any arbitrary node from $\mathcal{G}=(V,E)$.
\label{cor: MPS_conn_maint}
\end{corollary}
\begin{proof}
Let $p_0 \in V$ be any arbitrary node that will be removed from $\mathcal{G}$. If $p_0$ is a leaf node, MPS stops at $p_0$, and the connectivity maintenance is an immediate result. Otherwise, $p_0$ generates a message as $\{p_0\}$, and the message is modified as $\{p_0,...,p_i\}$ whenever it is received by $p_i \in V$. Let $N+1$ be the total number of nodes in $\mathcal{G}$. In this respect, as $i \rightarrow N$, $\{p_0,...,p_i\} \rightarrow \{p_0,...,p_N\} = V$. Eventually, there exist an instant $k=i \leq N$, at which $\mathcal{N}_{p_k} \subseteq \{p_0,...,p_k\}$. From Theorem~\ref{thm: mainThm}, $p_k$ is a noncritical node because it satisfies $\mathcal{N}_{p_k} \setminus \{p_0,...,p_k\} = \emptyset$. Consequently, MPS always stops at a noncritical node. Moreover, from Corollary~\ref{cor: guarantee_connectivity} the replacements based on MPS always guarantee connectivity maintenance because the graph is reconfigured as if $p_k$ is removed from $\mathcal{G}$ instead of $p_0$.
\end{proof}

An illustration for MPS is displayed in Figure~\ref{fig: MPSdemo}, where there is an initially connected graph with 7 nodes. As it is seen from Figure~\ref{fig: MPSdemo}(b), the removal of $v_0$ will create a disconnection in the graph. If each node runs MPS, then a replacement path is generated as $\{v_0,v_2,v_4\}$ such that $v_2$ replaces $v_0$, and $v_4$ replaces $v_2$. Note that $\{v_0,v_2,v_4\}$ is not the only replacement path, i.e. $\{v_0,v_1,v_5\}$. Consequently, the system reconfigures itself to maintain connectivity, and, in the resulting configuration, it is guaranteed to observe the removal of a noncritical node (e.g. $v_4$) instead of the removal an arbitrarily removed node (e.g. $v_0$).   

\begin{figure}[h!]
\begin{center}
\resizebox*{\columnwidth}{!}{\includegraphics{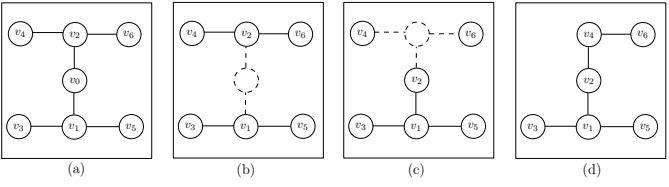}}%
\caption{An illustration for MPS. (a) Initially connected graph. (b) $v_0$ will leave the system. Since it is not a leaf node, it creates a ${\it message}$ as $\{v_0\}$ and selects a neighbor from $\mathcal{N}_{v_0} \setminus \{v_0\} = \{v_1,v_2\}$ to replace itself. (c) $v_2$ receives the ${\it message}$ and modifies it as $\{v_0,v_2\}$. Then, it selects a neighbor from $\mathcal{N}_{v_2} \setminus \{v_0,v_2\} = \{v_4,v_6\}$ to replace itself. (d) $v_4$ receives the ${\it message}$ and modifies it as $\{v_0,v_2,v_4\}$. It attempts to select a neighbor from $\mathcal{N}_{v_4} \setminus \{v_0,v_2,v_4\} = \emptyset$ for its replacement. Since $\mathcal{N}_{v_4} \setminus \{v_0,v_2,v_4\} = \emptyset$, $v_4$ cannot send the message to any node and the algorithm stops.}
\label{fig: MPSdemo}
\end{center}
\end{figure}

\subsection{Performance of MPS}
Given a networked system, reactive schemes for connectivity maintenance result in some changes in the graph topology. While maintaining the graph connectivity, an important aspect is not to cause significant changes in the graph properties such as the total number of edges or the maximum node degree. Note that the total number of edges and the maximum node degree can be directly related to the overall communication cost, whose increase is not desirable for a networked system containing agents with limited power capacity. 

\begin{proposition}
A sequence of replacements along a maximal simple path, $(p_0,p_1,...,p_k)$, on $\mathcal{G}$, such that every $p_{i+1} \in \mathcal{N}_{p_i} \setminus \{p_0,p_1,...,p_i\}$, guarantees no increase in the total number of edges and maximum node degree in the presence of any arbitrary node removal.   
\label{cor: guarantee_edge_degree}
\end{proposition}

\begin{proof}
Let $p=(p_0,p_1,...,p_k)$ be the replacement path, $\mathcal{G}^*$ be the new graph structure after the replacements. Then, this corollary is proven in two parts: (1) In the removal of an arbitrary node, $p_0$, $p$ results in $\mathcal{G}^*$, which corresponds to the removal of $p_k$ and its adjacent edges from $\mathcal{G}$. As a result, the total number of edges decrease as the agents are removed. (2) Let $p_0$ in $p$ be the agent that has the maximum degree $d_{max}$ in $\mathcal{G}$. If $p_0$ is removed, then $p_{1}$ replaces $p_0$. Now, if $k = 1$, then $p_1$ is the noncritical node that will not be replaced. As a consequence, the degree of $p_{1}$ becomes $d_{max}-1$ after the replacement. If $k \neq 1$, then $p_{1}$ will be replaced by $p_{2}$. Hence, the degree of $p_{1}$ becomes $d_{max}$ after the replacements. In both cases, $p_{1}$ becomes the node with the maximum degree in $\mathcal{G}^*$ after replacing $p_0$. Finally, in the removal of an arbitrary node, which does not correspond to the maximum degree node $\tilde{v}$, either no replacements occur in the neighborhood of $\tilde{v}$, or the replacements in the neighborhood of $\tilde{v}$ may cause at most one reduction in $d_{max}$. As a result, the maximum node degree in $\mathcal{G}^*$ becomes either $d_{max}$ or $d_{max}-1$.

\end{proof}

The optimal solution satisfying the minimum number of replacements for the RC problem can be obtained by a centralized controller by finding the shortest path between the removed node and a noncritical node on the graph. Note that such a centralized controller requires the complete information about the graph. The objective of MPS is to solve the RC problem only by using some local and partial information. Due to utilizing limited information, MPS may not necessarily guarantee the optimal solution for any graphs. In this section, we will discuss the performance of MPS for various graph structures.
 
 
\begin{proposition}
In any undirected connected graph, $\mathcal{G}=(V,E)$, the maximum number of replacements that can occur via MPS is $(|V|-1)$.  
\label{prop: anygraph_perf}
\end{proposition}
\begin{proof}
From Proposition~\ref{prop: MPS_max_simple_path}, MPS results in a message that is the sequence of replacements represented as a maximal simple path, $p$. Let $|p|\geq |V|+1$, then at least one node appears multiple times in $p$, thus $p$ is not simple. This is a contradiction, hence $|p|\leq|V|$ implying an upper bound for the number of replacements as $|V|-1$.    
\end{proof}

\begin{definition}
A tree graph is an undirected graph in which any two nodes are connected by exactly one simple path.
\label{def: tree}
\end{definition}

\begin{proposition}
In tree graphs, $\mathcal{G}=(V,E)$, the maximum number of replacements that can happen via MPS is $(\Delta-1)$, where $\Delta$ is the diameter of $\mathcal{G}$.  
\label{prop: treegraph_perf}
\end{proposition}
\begin{proof}
Note that a noncritical node in a tree graph is always a leaf node, and a critical node always has a degree of 2. In this manner, the diameter of a tree graph corresponds to the length of the longest maximal simple path. Let $\{p_0,p_1,....,p_{\Delta-1},p_{\Delta}\}$ denote to the longest maximal simple path. In this path, both $p_0$ and $p_{\Delta}$ are leaf nodes (noncritical), and the nodes in between are critical. If $p_0$ is the removed node, then MPS does not initiate replacements. If $p_1$ is the removed node, then the maximum number of replacements based on MPS may occur along the sequence $\{p_1,....,p_{\Delta-1},p_{\Delta}\}$ resulting in $(\Delta-1)$ replacements. 
\end{proof}

\begin{definition}
A biconnected graph is a connected graph that does not have any critical nodes.
\label{def: biconnected}
\end{definition}

\begin{proposition}
In biconnected graphs, MPS cannot achieve optimal solution for connectivity maintenance. 
\label{prop: cyclegraph_perf}
\end{proposition}
\begin{proof}
In a biconnected graph, $\mathcal{G}=(V,E)$, each node is noncritical. However, based on MPS, any arbitrary node, $v \in V$, that will be removed from $\mathcal{G}$ always initiates the replacements. In this manner, the graph is reconfigured with $|V|-1$ node replacements even though the removal of $v$ does not cause any disconnection in $\mathcal{G}$.
\end{proof}

Note that MPS may not always result in the minimum number of replacements in agent removal. For instance, if the removed agent is not a leaf node, but noncritical, MPS still initiates the sequence of replacements as depicted in Proposition~\ref{prop: cyclegraph_perf}. From a centralized perspective, understanding the criticality of a node is feasible. However, the node criticality may not be determined locally. As shown in Figure~\ref{fig: inf_graphs}, let $\mathcal{G}_1$ and $\mathcal{G}_2$ be an infinite cycle and infinite path graphs, respectively. Suppose that any node in $\mathcal{G}$ knows its $\delta$-hop neighborhood. Let $v_0$ be the removed node. As seen from Figure~\ref{fig: inf_graphs}, $v_0$ is noncritical in $\mathcal{G}_1$, but critical in $\mathcal{G}_2$. Note that for any finite $\delta$, $v_0$ has the same neighborhood in $\mathcal{G}_1$ and $\mathcal{G}_2$, hence it can not differentiate its criticality by just looking at its $\delta$-neighborhood. 

%

\begin{figure}[h!]
\begin{center}
\resizebox*{\columnwidth}{!}{\includegraphics{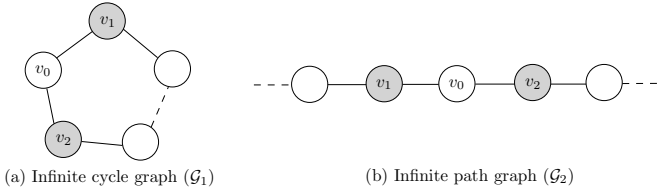}}%
\caption{Examples to graphs with infinite nodes.}%
\label{fig: inf_graphs}
\end{center}
\end{figure}

Since the criticality cannot be always determined locally, MPS may sometimes initiate a sequence of replacements when a noncritical node is removed. Such an optimality gap is due to the limitation of local information in the computations. However, it is important to emphasize that, for any undirected connected graph, connectivity maintenance in the presence of any node removal is guaranteed by MPS by using only some local information. As depicted, MPS always stops at a noncritical node. Thus, the graph reconfigures itself as if a noncritical node is removed from the system instead of an arbitrary node, $p_0$. In this respect, $p_{i}$ selects the node, $p_{i+1}$, which will replace itself, from $\mathcal{N}_{p_i} \setminus \{p_0,....,p_i\}$. Here, a question arises as which node from $\mathcal{N}_{p_i} \setminus \{p_0,....,p_i\}$ should be selected to increase the efficiency of MPS. For instance, a random selection scheme requires very little information to be shared among nodes, or a node selection based on the minimum degree, \cite{aksaray2013}, may capture the leaf node neighbors. Consequently, as the information possessed by a node and shared in the neighborhood increases, the solution approaches the optimal solution.

\section{$\delta$-criticality MPS} \label{deltaMPS}
In this section, we introduce a variant of MPS, which uses $\delta$-hop criticality information for each node. Here, the $\delta$-hop criticality is defined as Definition~\ref{def: delta-hopCriticality}.

\begin{definition} ($\delta$-hop criticality)
A node, $v_i$, is $\delta$-hop critical if the subgraph, induced by the $\delta$-neighborhood of $v_i$, is disconnected by the removal of $v_i$; otherwise, $v_i$ is $\delta$-hop noncritical.
\label{def: delta-hopCriticality}
\end{definition}

\begin{remark}
Let $\mathcal{G}=(V,E)$ be a connected graph, and let $v_i \in V$ be $\delta$-hop noncritical. Then, $v_i$ is noncritical in $\mathcal{G}$. Suppose that a simple path, $p^{nm}$, connects any arbitrary two nodes $v_n, v_m \in V$ and includes $v_i$ as an intermediate node. In $p^{nm}$, $v_i$ appears between two of its neighbors. In the removal of $v_i$, there exist another path, $p^i$, consisting of some nodes within $\delta$ hops of $v_i$ since $v_i$ is $\delta$-hop noncritical by definition. Hence, the removal of $v_i$ does not cause a disconnection between $v_n$ and $v_m$ because $v_i$ can be replaced by $p^i$. Consequently, a $\delta$-hop noncritical node is always a noncritical node in $\mathcal{G}$.  
\label{rem: delta-hop-NonCritical}
\end{remark}

In light of Remark~\ref{rem: delta-hop-NonCritical}, $\delta$-hop criticality is used in MPS as in Algorithm~2. In this respect, each node knows whether itself and immediate neighbors are $\delta$-hop critical.

\begin{center}
\resizebox{\columnwidth}{!}{
\begin{tabular}{l l}
\hline
\bf{Agorithm 2:} \bf$\delta$-criticality MPS\\
\hline
$Input :$ An arbitrary node, $p_0$, from $\mathcal{G}$ \\
$Output :$ Connectivity maintenance in the removal of $p_0$ \\
$Assumption:$ Each node shares both its unique node ID and $\delta$-criticality with its neighbors.\\
\mbox{\small $\;1:\;$}\textbf{initialization:} $p_i \gets p_0; \quad \quad \mathcal{N}_{p_i} \gets \mathcal{N}_{p_0}; \quad$ $message \gets (p_0)$;\\
\mbox{\small $\;2:\;$}\textbf{if}\hspace{0.1cm}  $| \mathcal{N}_{p_0} |=1$ \\
\mbox{\small $\;3:\;$}\hspace{0.45cm} no replacements required; \\
\mbox{\small $\;4:\;$}\textbf{else}\hspace{0.1cm}  $ $ \\
\mbox{\small $\;5:\;$}\hspace{0.45cm} \textbf{while} \hspace{0.1cm} $\mathcal{N}_{p_i} \setminus message \neq \emptyset$\\
\mbox{\small $\;7:\;$}\hspace{0.90cm} \textbf{if} \hspace{0.1cm} any $v \in \mathcal{N}_{p_i} \setminus message$ is $\delta$-noncritical;\\
\mbox{\small $\;8:\;$}\hspace{1.35cm} $p_{i+1} \gets v$ s.t. $v$ is one of the $\delta$-noncritical nodes; \\
\mbox{\small $\;9:\;$}\hspace{0.90cm} \textbf{else} \\
\mbox{\small $\;10:\;$}\hspace{1.35cm} $p_{i+1} \gets v$ s.t. $v$ is randomly selected from  $\mathcal{N}_{p_i} \setminus message$; \\
\mbox{\small $\;11:\;$}\hspace{0.71cm} \textbf{end if} \\
\mbox{\small $\;12:\;$}\hspace{0.71cm} $p_i$ sends $message$ to $p_{i+1}$;\\
\mbox{\small $\;13:\;$}\hspace{0.71cm} $p_i$ replaces the second last node in the $message$; \\
\mbox{\small $\;14:\;$}\hspace{0.71cm} $p_i \gets p_{i+1}$; \quad $\mathcal{N}_{p_i} \gets \mathcal{N}_{p_{i+1}}$; \\ 	
\mbox{\small $\;6:\;$}\hspace{0.90cm} $message \gets (message , p_i)$ ;\\
\mbox{\small $\;15:\;$}\hspace{0.45cm} \textbf{end while} \\
\mbox{\small $\;16:\;$}\textbf{end if}\hspace{0.1cm}  $ $ \\
\hline
\end{tabular}
}
\end{center}
\vskip1ex
In $\delta$-criticality MPS, whenever a node, $p_i$, receives a message, it adds its own individual ID likewise MPS. Then, it selects a neighbor from the candidate set, $\mathcal{N}_{p_i} \setminus \{p_0,...,p_i\}$, based on $\delta$-criticality. In the case, where the candidate set does not contain a $\delta$-hop noncritical node, $p_i$ selects a random node from the candidate set. 



It has been shown in Remark~\ref{rem: delta-hop-NonCritical} that a $\delta$-noncritical node is globally noncritical in $\mathcal{G}$. Now, a question arises as when a $\delta$-critical node assures global criticality. In this respect, Proposition~\ref{thm: deltaCritical} presents a sufficient condition that guarantees global node criticality by relating $\delta$ to a graph structure.

\begin{definition}
A chordless cycle in $\mathcal{G}$ is a cycle such that no two nodes of the cycle are connected by an edge that does not itself belong to the cycle.
\label{def: chordlessCycle}
\end{definition}

\begin{proposition}
Let $c_{max}$ be the length of the longest chordless cycle in $\mathcal{G}$. If $\delta \geq \frac{c_{max}}{2}$, then a $\delta$-critical node is globally critical in $\mathcal{G}$.
\label{thm: deltaCritical}
\end{proposition}

\begin{proof}
Let $v$ be a noncritical node in $\mathcal{G}$, and let $\mathcal{N}^{\delta}$ be the $\delta$-neighborhood of $v$ for some $\delta \geq \frac{c_{max}}{2}$, where $c_{max}$ is the length of the longest chordless cycle in $\mathcal{G}$. Suppose that $v$ is a $\delta$-critical node, then the graph, $\mathcal{G}'$, induced by the nodes in $\mathcal{N}^{\delta}$ is disconnected. Now, since $v$ is noncritical, there exist a shortest path between the nodes $(u,w)\in \mathcal{N}^{\delta}$, which are not connected in $\mathcal{G}'$ but connected in $\mathcal{G} - v$. Moreover, there always exist a shortest path, $(u,p^*,w)$, where no elements on $p^*$ is connected to $v$ (in other words, no elements on $p^*$ is in $\mathcal{N}^{\delta}$). Note that $(u,p^*,w,v,u)$ is a chordless cycle and its length, $c$, cannot be larger than $c_{max}$. However, $v$ does not know the existence of such a path, so $c > 2\delta$, which is a contradiction because $2\delta \geq c_{max} \geq c$.
\end{proof}

\begin{corollary}
If $\delta \geq \frac{c_{max}}{2}$, then the replacement sequence generated via $\delta$-MPS involves only one noncritical node, which is the last node on the replacement sequence.
\label{cor: delta_1noncriticalnode}
\end{corollary}

\begin{proof}
Based on Algorithm~2, a message travels from a $\delta$-critical node to a neighboring $\delta$-critical node until finding a $\delta$-noncritical node. In the case of $\delta \geq \frac{c_{max}}{2}$, Proposition~\ref{thm: deltaCritical} shows that a $\delta$-critical node is globally critical. Hence, the replacement sequence generated via $\delta$-MPS contains only one noncritical node, which is the last node on the sequence. 
\end{proof}

\begin{remark}
Suppose that there is no chordless cycle in $\mathcal{G}$. A $\delta$-critical node for any $\delta \geq 1$ is globally critical in $\mathcal{G}$ because $\mathcal{G}$ is a tree graph where each noncritical node is a leaf node. 
\label{rem: diam-1}
\end{remark}

\begin{remark}
If $\delta$ is selected properly based on the graph topology (based on Proposition~\ref{thm: deltaCritical} and Remark~\ref{rem: diam-1}), a $\delta$-critical node is always a critical node. Hence, a resulting replacement sequence does not contain any redundant replacements since only the final node of the sequence is non-critical.

\end{remark}

As it is seen, if $\delta$ is selected properly such that a $\delta$-critical node is globally critical, then $\delta$-MPS does not cause any unnecessary agent replacements to maintain connectivity, and the resulting sequence approaches the optimal solution. Note that avoiding any unnecessary replacements is crucial for networked systems with limited power capacity. In some graph structures, there might not necessarily exist a unique replacement sequence. As such, if $\delta$-noncritical nodes are beyond the immediate neighborhood of a node $v$, then $v$ selects a $\delta$-critical neighbor randomly for its replacement. Due to the randomized nature of selecting replacing agent, $\delta$-MPS may not always guarantee the shortest path to a noncritical node. 

In Figure~\ref{fig: random_dMPS}, a line graph involving $7$ nodes is presented. In this example, let $\delta=1$, then $v_1,v_0,v_3,v_4,v_5$ are 1-critical nodes whereas $v_2$ and $v_6$ are noncritical nodes. Since the graph does not contain a chordless circle, from Remark~\ref{rem: diam-1}, a 1-critical node is globally critical. Here, we illustrate that $\delta$-MPS may result in a sequence of necessary replacements, but a longer route, for connectivity maintenance. In this manner, assume that $v_0$ is removed from the graph. From the perspective of $v_0$, selecting $v_1$ or $v_3$ is indifferent because $v_0$ can only see $\mathcal{N}^{\delta=1}$, which contains the highlighted nodes. Hence, $\delta$-MPS results in either $(v_0,v_1,v_2)$ or $(v_0,v_3,v_4,v_5,v_6)$ as a replacement sequence.  

\begin{figure}[h!]
\begin{center}
\resizebox*{\columnwidth}{!}{\includegraphics{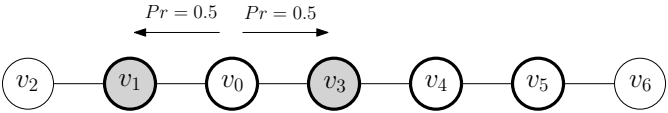}}%
\caption{In the case of $v_0$ is removed, $1$-MPS generates either $(v_0,v_1,v_2)$ or $(v_0,v_3,v_4,v_5,v_6)$ as a replacement sequence.}
\label{fig: random_dMPS}
\end{center}
\end{figure}   

Consequently, the optimality of $\delta$-MPS depends on the lengths of the shortest paths from a removed node to one of the noncritical nodes. If the lengths does not significantly vary from each other, then the solution of $\delta$-MPS is close to the optimal one.

\section{Simulation Studies} \label{analysis}
In order to elaborate on the performance of the message passing strategy, the Monte Carlo simulations are conducted to understand how close MPS is to the centralized solution. For the simulations, we use the MATLAB simulation environment and consider the canonical scenario depicted in the following section.

\subsection{Canonical Scenario} 
Consider an indoor reconnaissance and surveillance mission, where a set of robots (i.e, small unmanned vehicles) gather data from critical points and share the gathered data to increase their situational awareness. In such a mission, an undirected connected graph, $\mathcal{G}$, associates with the critical viewpoints of the environment (i.e. the viewpoints are the feasible points that a robot can be located on) as in Figure~\ref{fig: scenario}, where the nodes are the viewpoints and the dashed edges represent the communication links if two robots are located on the corresponding nodes. 

\begin{figure}[h!]
\begin{center}
\resizebox*{6cm}{!}{\includegraphics{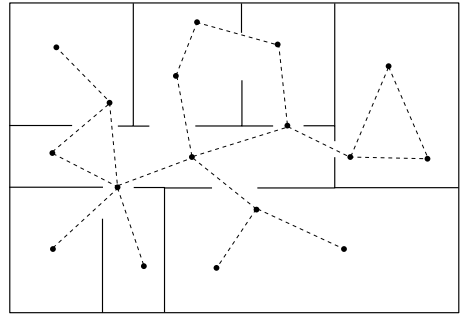}}%
\caption{An indoor environment associated with an undirected connected graph, whose nodes represent the critical viewpoints and edges represent the possible communication capability if agents are located on the corresponding nodes.}%
\label{fig: scenario}
\end{center}
\end{figure}

In this scenario, let each robot have limited energy capacity and different energy consumption. Then, it is likely to observe that the robots have variable energy levels. In this respect, a robot leaves the group when it reaches an energy threshold. Here, the removal of an agent may cause a disconnection in the communication network. For example, let 13 agent be assigned to points such that a connected communication network is obtained as in (a) of Figure~\ref{fig: scenario_combined}. Then the removal of $v_4$ causes a disconnection in the communication network as in (b) of Figure~\ref{fig: scenario_combined}. Consequently, the objective in this scenario is to maintain a connected communication network among the remaining agents with minimum agent replacements. 

\begin{figure}[h!]
\begin{center}
\resizebox*{\columnwidth}{!}{\includegraphics{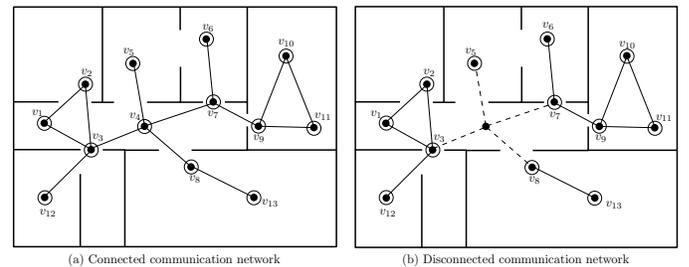}}%
\caption{(a) Thirteen agents assigned to the feasible tasks and having a connected communication network. (b) The removal of $v_4$ causes four partitioning in the communication network.}
\label{fig: scenario_combined}
\end{center}
\end{figure}

\subsection{Monte Carlo Simulations}
Based on the canonical scenario, the objective of the Monte Carlo simulations is to provide a statistical understanding for the optimality of MPS and $\delta$-criticality MPS with respect to the centralized solution. In all simulations, we consider $20$ assignments, which associate with a randomly generated undirected connected graph. Initially, we assume $20$ agents individually assigned to each node. At an instant, a randomly selected arbitrary agent is removed from the network. To maintain the connectivity, we solve the RC problem via centralized controller, MPS, and $\delta$-criticality MPS.      

The results of the Monte Carlo simulations show that the optimality of MPS significantly improves as $\delta$-criticality information is incorporated to the decision scheme for the replacements. For example, as illustrated in Table~\ref{tab: results}, an increase in $\delta$ leads to the total number of replacements induced from $\delta$-criticality MPS to approach the number of replacements resulted from the centralized solution. Note that the huge cost difference between MPS and $\delta$-criticality MPS is due to the fact that MPS is using a randomized decision mechanism for the replacements while $\delta$-criticality MPS incorporates $\delta$-criticality of a node in the decision mechanism. Furthermore, the optimality of MPS also varies with respect to the graph topology. For example, in the case of the graph diameter decreases and the average node degree increases, the performance of MPS degrades while the performance of $\delta$-criticality MPS improves greatly. 

\begin{table}[h!]
	\center
  \caption{Based on 500 simulations, mean cost$^\ast$ of various strategies for connectivity maintenance in graphs with 20 nodes}
\resizebox{\columnwidth}{!}{  
{\begin{tabular}{@{}cccccc}\hline
   	   
   	Mean diameter & Mean avg. node degree 
   		& centralized & MPS 
        & 1-criticality MPS
        & 2-criticality MPS\\
\hline
 	$9.2$ & $3.1$ & $0.498$ & $4.960$ & $0.626$ & $0.556$ \\
 	$6.7$ & $4.3$ & $0.186$ & $7.224$ & $0.302$ & $0.240$ \\
 	$4.0$ & $7.3$ & $0.028$ & $12.574$ & $0.068$ & $0.030$ \\
   \hline
   $^\ast${\tiny number of agent replacements}
  \end{tabular}}
}  
  \label{tab: results}
\end{table}

\begin{figure}[h!]
\begin{center}
\resizebox*{\columnwidth}{!}{\includegraphics{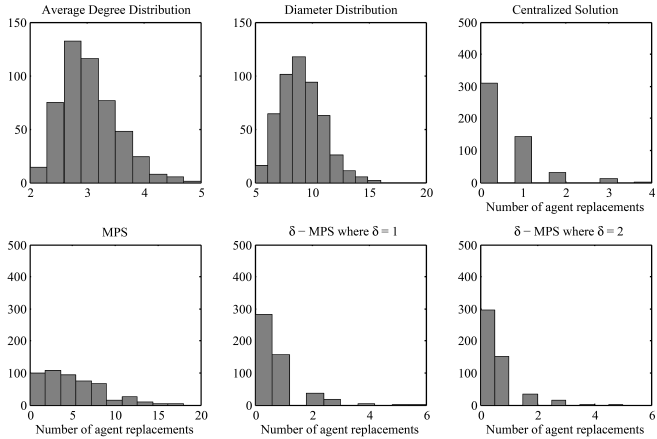}}%
\caption{Distributions pertaining to the graph properties and replacement solutions. As $\delta$- criticality is utilized, the number of replacements for preserving connectivity approaches to the number induced by a centralized controller.}
\label{fig: results1}
\end{center}
\end{figure}

\begin{figure}[h!]
\begin{center}
\resizebox*{\columnwidth}{!}{\includegraphics{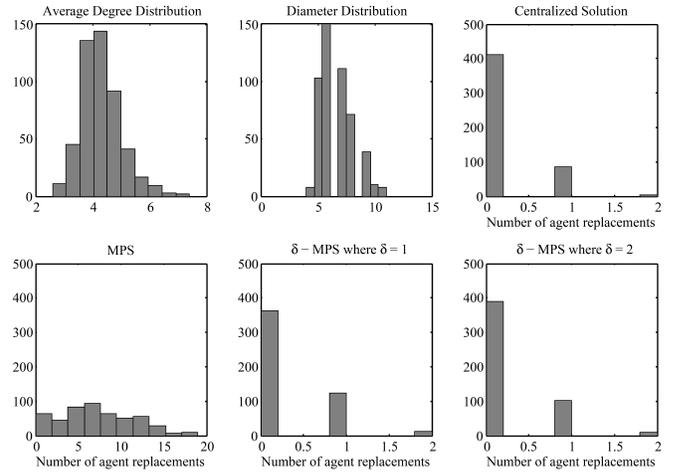}}%
\caption{Distributions pertaining to the graph properties and replacement solutions. As the graph diameter decreases and the average node degree increases, the replacement sequence driven by MPS diverges from the centralized solution.}
\label{fig: results2}
\end{center}
\end{figure}

\begin{figure}[h!]
\begin{center}
\resizebox*{\columnwidth}{!}{\includegraphics{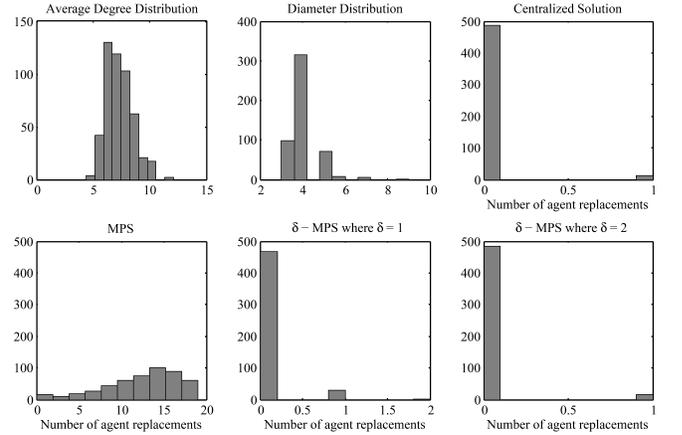}}%
\caption{Distributions pertaining to the graph properties and replacement solutions. As the graph diameter decreases and the average node degree increases, the performance of $\delta$-criticality MPS improves greatly and converges to the performance of centralized controller.}
\label{fig: results3}
\end{center}
\end{figure}

\section{Conclusions} \label{conclusion}
In this study, the connectivity issue of networked systems in the presence of agent removal has been discussed, and a decentralized connectivity maintenance strategy, which is applicable to any scale of network, has been proposed. We showed that the message passing strategy (MPS) proposed in this paper maintains the graph connectivity for any initially connected network until the removal of all agents. This is achieved by a sequence of replacements initiated by the removed agent. The benefits of the proposed control scheme are guaranteeing the connectivity maintenance by using only some local information and not increasing the total number of edges and the maximum node degree of a network as the agents are removed. 

The optimality gap of the proposed strategy has been discussed through Monte Carlo simulations by comparing the performance of the proposed decentralized strategy with respect to the centralized solution, which results in the minimum number of replacements. While the message passing strategy maintains the graph connectivity even in the case of replacements by randomly selected agents, it has been observed that incorporating $\delta$-criticality information to the decision mechanism significantly improves the resulting performance. As such, a variant of MPS has been introduced as $\delta$-criticality MPS, which demonstrates a significant performance improvement. 

Some interesting aspects requiring further investigation include the development of a strategy for simultaneous agent removals, a throughout study for a more general communication problem involving the delays and uncertainty, and the introduction of a more general mission, in which some of the removed agents return to the mission area.

\ifCLASSOPTIONcaptionsoff
  \newpage
\fi

\bibliography{/PhD_Research/References/ref2}
%

%
%
%




\end{document}